\documentclass{llncs}
\usepackage{amsmath}
\usepackage{amssymb}
\usepackage{url}
\usepackage{enumerate}
\usepackage{boxedminipage}

\begin{document}

\title{Parameterized Study of the Test Cover Problem}

\author{R. Crowston\inst{1} \and G. Gutin\inst{1} \and M. Jones\inst{1} \and  S. Saurabh\inst{2}  \and  A. Yeo\inst{3}}

\institute{Royal Holloway, University of London\\
Egham TW20 0EX, UK, \email{\{robert|gutin|markj\}@cs.rhul.ac.uk}
\and The Institute of Mathematical Sciences\\
Chennai 600 113, India, \email{saket@imsc.res.in}
\and University of Johannesburg\\
Auckland Park, 2006 South Africa, \email{anders.yeo.work@gmail.com}
}
\date{}
\maketitle

\begin{abstract}
In this paper we carry out a systematic study of a natural covering problem, used for identification
across several areas, in the realm of parameterized complexity. In the {\sc Test Cover} problem
we are given a set  $[n]=\{1,\ldots,n\}$ of items together with a collection, $\cal T$, of distinct subsets
of these items called tests. We assume that $\cal T$ is a test cover, i.e., for
each pair of items there is a test in $\cal T$ containing exactly one of these items.
The objective is to find a minimum size subcollection of $\cal T$, which is still a test cover.
The generic parameterized version of
{\sc Test Cover} is denoted by $p(k,n,|{\cal T}|)$-{\sc Test Cover}. Here, we are given  $([n],\cal{T})$ and a positive integer parameter
$k$ as input and the objective is to decide whether there is a test cover of size at most $p(k,n,|{\cal T}|)$.
We study four parameterizations for {\sc Test Cover} and obtain the following:

(a) $k$-{\sc Test Cover}, and $(n-k)$-{\sc Test Cover} are fixed-parameter tractable (FPT), i.e., these problems
 can be solved by algorithms of runtime  $f(k)\cdot poly(n,|{\cal T}|)$, where $f(k)$ is a function of $k$ only.

(b) $(|{\cal T}|-k)$-{\sc Test Cover} and $(\log n+k)$-{\sc Test Cover} are W[1]-hard.
Thus, it is unlikely that these problems are FPT.

\end{abstract}

\section{Introduction}\label{sec:i}

The input to the {\sc Test Cover} problem consists of a set of items, $[n]:=\{1,2, \ldots ,n\}$,
and a collection of distinct sets, ${\cal T}=\{T_1,\ldots , T_m\}$, called {\em tests}. We say that a test
$T_q$ {\em separates} a pair $i,j$ of items if $|\{i,j\}\cap T_q|=1.$ Subcollection ${\cal T}'\subseteq {\cal T}$
is a {\em test cover} if each pair $i,j$ of distinct items is separated by a test in ${\cal T}'$.
The objective is to find a test cover of minimum size, if one exists. Since it is easy to decide,
in polynomial time, whether the collection ${\cal T}$ itself is a test cover,
henceforth we will assume that ${\cal T}$ is a test cover.

{\sc Test Cover} arises naturally in the following general setting of
identification problems: Given a set of items and a set
of binary attributes that may or may not occur in each item, the aim
is to find the minimum size subset of attributes (corresponding to a minimum test cover) such that each
item can be uniquely identified from the information on which of this subset
of attributes it contains.
{\sc Test Cover} arises in fault analysis, medical diagnostics, pattern recognition, and
biological  identification (see, e.g., \cite{HalHalRav01,HalMinRav01,MorSha85}).

The {\sc Test Cover} problem has been also studied extensively from an algorithmic view point.
The problem is NP-hard, as was shown by Garey and Johnson \cite{GarJoh79}.
Moreover,  {\sc Test Cover} is APX-hard \cite{HalHalRav01}. There is an $O(\log n)$-approximation algorithm for
the problem \cite{MorSha85} and there is no $o(\log n)$-approximation algorithm unless P=NP \cite{HalHalRav01}.
These approximation results are obtained using reductions from {\sc Test Cover} to the well-studied {\sc Set Cover}
problem, where given a collection ${\cal S}$ of subsets of $[n]$ {\em covering} $[n]$ (i.e., $\cup_{X\in \cal S}X=[n]$)
and integer $t$, we are to decide whether there is a subcollection of $\cal S$ of size $t$ covering $[n]$.

In this paper we carry out a systematic study of  {\sc Test Cover} in the realm of
parameterized complexity\footnote{Basic notions on parameterized complexity are given in the end of this section.}.
The following will be a generic parameterization of the problem:

\begin{center}
\fbox{~\begin{minipage}{0.9\textwidth}
$p(k,n,m)$-{\sc Test Cover}\\ \nopagebreak
  \emph{Instance:} A set ${\cal T}$ of $m$ tests on $[n]$ such that $\cal T$ is a test cover.\\
    \nopagebreak
  \emph{Parameter:} $k$.\\ \nopagebreak
  \emph{Question:} Does ${\cal T}$ have a test cover with at most $p(k,n,m)$ tests?
\end{minipage}~}
\end{center}

We will consider four parameterizations of {\sc Test Cover}: $k$-{\sc Test Cover},  $(m-k)$-{\sc Test Cover}, $(n-k)$-{\sc Test Cover}, and $(\log n+k)$-{\sc Test Cover}. The first parameterization is standard; its complexity is not hard to establish
 since it is well-known
that there is no test cover of size less than $\lceil \log n\rceil$ \cite{HalHalRav01} and the bound is tight. This  bound suggests the parameterization $(\log n+k)$-{\sc Test Cover} (above a tight lower bound). The parameterization $(m-k)$-{\sc Test Cover} is a natural parameterization below a tight upper bound.
There is always a test cover of size at most $n-1$ \cite{Bon72} and ${\cal T}=\{\{1\},\ldots ,\{n-1\}\}$ shows that the bound is tight. Thus, $(n-k)$-{\sc Test Cover} is  another parameterization below a tight upper bound.

In this paper, we will use some special cases of the following generic parameterization of {\sc Set Cover}:

\begin{center}
\fbox{~\begin{minipage}{0.9\textwidth}
$p(k,n,m)$-{\sc Set Cover}\\ \nopagebreak
  \emph{Instance:} A collection ${\cal S}$ of $m$ subsets of $[n]$ covering $[n]$.\\
    \nopagebreak
  \emph{Parameter:} $k$.\\ \nopagebreak
  \emph{Question:} Does ${\cal S}$ contain a subcollection  of size $p(k,n,m)$ covering $[n]$?
\end{minipage}~}
\end{center}

Three of our parameterizations for {\sc Test Cover} are below or above guaranteed lower or upper bounds.
The study of parameterized problems above a guaranteed lower/upper bound was initiated by Mahajan and Raman~\cite{MahajanR99}.
They showed that some above guarantee versions of {\sc Max Cut} and {\sc Max Sat} are FPT; in the case of {\sc Max Sat}
the input is a CNF formula with $m$ clauses together with an integer $k$ (the parameter)
and the question is whether there exists an assignment that satisfies at least $m/2+k$ clauses.
 Later, Mahajan et al.~\cite{MahajanRS09} published  a paper with several new
results and open problems around parameterizations beyond guaranteed lower and upper bounds. In a breakthrough paper Gutin et al.~\cite{GutinKSY11}
developed a probabilistic approach to problems parameterized above or below tight bounds. Alon et al.~\cite{AlonGKSY10} combined this approach
with a method from Fourier analysis to obtain an FPT algorithm for parameterized {\sc Max $r$-SAT} beyond
the guaranteed lower bound. In the same paper a quadratic kernel was also given for {\sc Max $r$-SAT}.
Other significant results in this direction include quadratic kernels for ternary
permutation constraint satisfaction problems parameterized above average and results on systems of linear equations
over the field of two elements~\cite{CrowstonGJY11,CrowstonGJKR10,GutinIMY10}.

We establish parameterized complexity of all four parameterizations of {\sc Test Cover}:

 (i) Since there is no test cover of size less than $\lceil \log n\rceil$, $k$-{\sc Test Cover} is FPT: if $k< \log n$, the answer for $k$-{\sc Test Cover} is {\sc No} and, otherwise, $n\le2^k$ and $m\le 2^n\le2^{2^k}$ and so we can solve $k$-{\sc Test Cover} by brute force in time dependent only on $k.$

(ii) In Section \ref{sec:m-k}, we provide a polynomial-time reduction from the {\sc Independent Set} problem to $(m-k)$-{\sc Test Cover} to show that $(m-k)$-{\sc Test Cover} is W[1]-hard. A reduction from $(m-k)$-{\sc Set Cover} to $(m-k)$-{\sc Test Cover} and a result from \cite{GutJonYeo11} allows us to conclude that $(m-k)$-{\sc Test Cover} is W[1]-complete. Thus, $(m-k)$-{\sc Test Cover} is not fixed-parameter tractable unless FPT=W[1].

(iii) In Section \ref{sec:n0k}, we prove the main result of this paper: $(n-k)$-{\sc Test Cover} is FPT.
  The proof is quite nontrivial and utilizes a ``miniaturized" version of $(n-k)$-{\sc Test Cover} introduced and studied in Subsection \ref{sec:mini}.

(iv) Moret and Shapiro \cite{MorSha85} obtained a polynomial-time reduction from {\sc Set Cover} to {\sc Test Cover} such that the
{\sc Set Cover} problem has a solution of size $k$ if and only if its reduction to {\sc Test Cover} has a solution of size
$k+\lceil \log n\rceil.$ Since $k$-{\sc Set Cover} is W[2]-complete \cite{DowneyFellows99},
we conclude that $(\log n+k)$-{\sc Test Cover} is W[2]-hard. Thus, $(\log n+k)$-{\sc Test Cover} is not fixed-parameter tractable provided FPT$\neq$ W[2].\\

%
%

\noindent{\bf Basics on Parameterized Complexity.}
A parameterized problem $\Pi$ can be considered as a set of pairs
$(I,k)$, where $I$ is the \emph{problem instance} and $k$ (usually a nonnegative
integer) is the \emph{parameter}.  $\Pi$ is called
\emph{fixed-parameter tractable } if membership of $(I,k)$ in
$\Pi$ can be decided by an algorithm of runtime $O(f(k)|I|^c)$, where $|I|$ is the size
of $I$, $f(k)$ is an arbitrary function of the
parameter $k$ only, and $c$ is a constant
independent from $k$ and $I$. The class of fixed-parameter tractable problems is denoted by FPT.

When the decision time is replaced by the much more powerful $O(|I|^{f(k)}),$
we obtain the class XP, where each problem is polynomial-time solvable
for any fixed value of $k.$ There is an infinite number of parameterized complexity
classes between FPT and XP (for each integer $t\ge 1$, there is a class W[$t$]) and they form the following tower:
$FPT \subseteq W[1] \subseteq W[2] \subseteq \cdots \subseteq W[P] \subseteq XP.$
For the definition of classes W[$t$], see, e.g., \cite{DowneyFellows99,FlumGrohe06}. It is well-known that FPT$\neq$XP and it is widely believed
that already FPT$\neq$W[1]. Thus, by proving that a problem is W[1]-hard, we essentially rule out that the problem is
fixed-parameter tractable (subject to FPT$\neq$W[1]). For more information on parameterized complexity, see monographs \cite{DowneyFellows99,FlumGrohe06,Niedermeier06}.

\section{Complexity of $(m-k)$-{\sc Test Cover}}\label{sec:m-k}

In this section we give the hardness result for $(m-k)$-{\sc Test Cover}.
\begin{theorem}
$(m-k)$-{\sc Test Cover} is W[1]-complete.
\end{theorem}
\begin{proof}
We will give a reduction from the W[1]-hard $k$-{\sc Independent Set}  problem  to $(m-k)$-{\sc Test Cover}. An input to
$k$-{\sc Independent Set} consists of an undirected graph $G=(V,E)$ and a positive integer $k$ (the parameter) and the objective is to decide
whether there exists an independent set of size at least $k$ in $G$. A set $I\subseteq V$ is {\em independent} if
no edge of $G$ has both end-vertices in $I$.

Let $G$ be an input graph to $k$-{\sc Independent Set} with vertices $v_1,\ldots ,v_p$ and edges $e_1,\ldots ,e_q$. We
construct an instance of  $(m-k)$-{\sc Test Cover} as follows. The set of items is $\{e_i,e'_i:\ i\in [q]\}$
and the collection of tests is
$\{T_j:j\in~[p]\} \cup~\{T'_i : i\in [q-1] \}$, where
$T_j=\{e_i :\ v_j \in e_i\}$,
the set of of edges of $G$ incident to $v_j$, and
$T'_i=\{e_i,e'_i\}$.

A set $U$ of vertices of $G=(V,E)$ is a {\em vertex cover} if every edge of $G$ has at least one end-vertex in $U$. It is well-known and easy to see
that $U$ is a vertex cover if and only if $V\setminus U$ is an independent set.
Consider a minimum size vertex cover $U$ of $G,$ and a
test subcollection $ \{T_j : v_j \in U \} \cup \{ T'_i : i\in [q-1]\}$.
Observe that the latter is a test cover, since a pair $e_i,e'_j$ $(i\neq j)$ is separated by $T'_{\min\{i,j\}}$,
as are the pairs $e_i, e_j$ and $e'_i, e'_j$,
and a pair $e_i,e'_i$ is separated by
$T_j$ for some $v_j \in U$ such that $v_j \in e_i$. Such a $v_j$ exists since $U$ is a vertex cover.

A test cover must use all $T'_i$ as otherwise we cannot separate $e'_i,e'_q$ for some $i\neq q$. A test cover must also
use at least $|U|$ of $T_j$ tests. Suppose not, and consider the corresponding set $W$ of vertices, such that $|W|<|U|$.
 Then
every $e_i$ is separated from $e'_i$ by $T_j$ for some $v_j \in W$, and so
$W$ forms a vertex cover,
contradicting the minimality of $U$. Hence $G$ has a vertex cover of size $t$ if and only if
there is a test cover of size $q-1+t$.

The number of tests is $M=q-1+p$, and so there is a test cover of size $M-k=q-1+p-k$ if and only if $G$ has
an independent set with at least $k$ vertices. Since $k$-{\sc Independent Set} is W[1]-hard, $(m-k)$-{\sc Test Cover} is
W[1]-hard as well.

To prove that $(m-k)$-{\sc Test Cover} is in W[1], we will use the following reduction of {\sc Test Cover} to {\sc Set Cover}
by Moret and Shapiro \cite{MorSha85}. Consider an instance of {\sc Test Cover} with set $[n]$ of items  and set ${\cal T}=\{T_1,\ldots ,T_m\}$ of tests.
The corresponding instance of {\sc Set Cover} has ground set $V=\{ (i,j):\ 1\le i<j\le n\}$ and set collection
 $\{S_q: q \in [m]\}$, where $S_q=\{(i,j)\in V:\ T_q \mbox{ separates } i,j\}$.
 Observe that the instance of {\sc Test Cover} has a test cover of size $\mu$ if and only if
the corresponding instance of {\sc Set Cover} has a set cover of size $\mu$. It is proved in \cite[Theorem 4]{GutJonYeo11} that $(m-k)$-{\sc Set Cover} is in W[1]. Hence,
$(m-k)$-{\sc Test Cover} is in W[1] as well. This completes the proof.
\qed \end{proof}

\section{Complexity of $(n-k)$-{\sc Test Cover}}\label{sec:n0k}
In this section we prove that $(n-k)$-{\sc Test Cover} is fixed-parameter tractable. Towards this we first introduce an
equivalence relation on $[n]$.

Given a subcollection ${\cal T}'\subseteq {\cal T}$, and two items $i,j\in [n]$, $i\neq j$ we write that
$i\equiv_{{\cal T}'} j$, if $i,j$ is not separated by any tests in ${\cal T}'$. Clearly, $\equiv_{{\cal T}'}$
is an equivalence relation on $[n]$. Essentially, each equivalence class is a maximal set $C \subseteq [n]$
such that no pair $i,j \in C$ is separated by a test in ${\cal T}'$; we say that $C$ is a
{\em class induced by } ${\cal T}'$. Observe that ${\cal T}'$ is a test cover if and only if each class induced by ${\cal T}'$ is a singleton, i.e.,
there are exactly $n$ classes induced by ${\cal T}'$.


\subsection{$k$-Mini Test Cover}\label{sec:mini}

To solve $(n-k)$-{\sc Test Cover} we first introduce a ``miniaturized'' version of the problem, namely,
the {\sc $k$-Mini Test Cover} problem. Here, we are given a set  $[n]$ of items
and a collection ${\cal T}=\{T_1,\ldots , T_m\}$ of tests.
As with {\sc Test Cover}, we assume that ${\cal T}$ is a test cover.
 We say that a subcollection ${\cal T}'\subseteq {\cal T}$ is a {\em $k$-mini test cover} if $|{\cal T}'|\le 2k$ and the number of classes induced by ${\cal T}'$ is at least $|{\cal T'}|+k$. We say a test $T$ {\em separates} a set $S$ if there exist $i,j \in S$ such that
$T$ separates $i,j$. Our main goal in this subsection is to show that   the $(n-k)$-{\sc Test Cover} problem and the
{\sc $k$-Mini Test Cover} problem are equivalent. Towards this we first show the following lemma.

\begin{lemma}\label{lem:oldClaim}
Suppose that ${\cal T}$ is a test cover for $[n]$, ${\cal F} \subseteq {\cal T}$ and the number of classes induced by
$\cal F$ is at least $|{\cal F}|+k$.
Then ${\cal F}$ can be extended to a test cover of size at most $n-k$.
Moreover, if ${\cal T}$ contains all singletons, this is possible by adding only singletons.
\end{lemma}

\begin{proof}
Add tests from ${\cal T}$ to $\cal F$ one by one such that each test increases the number of classes induced by $\cal F$, until the number of classes is $n$.
This can be done, since if we have less than $n$ classes, there is a class $C$ containing at least two items. For $i,j \in C$
there exists a test $T$ in ${\cal T} \setminus {\cal F}$ that separates $i,j$ which may be added to $\cal F$.
If we are only permitted to add singletons, then pick $T=\{i \}$.
Let ${\cal F}'$ be the subcollection produced from $\cal F$ in this way.
Observe that ${\cal F}'$ is a test cover. Since $\cal F$ induces at least $|{\cal F}| + k$ classes, we need to add at most $n-(|{\cal F}|+k)$ tests
to produce ${\cal F}'$. Thus $|{\cal F}'| \le n-k$,
as required.
\qed \end{proof}

We now define the  notion of a {\em $C$-test} as follows.

\begin{definition}
Let $C\subseteq [n]$. A test $S \in {\cal T}$ is a $C$-test if
$C \setminus S \not= \emptyset$ and $S \cap C \not= \emptyset$
(i.e. $S$ separates $C$).
 We also define the {\em local} portion of a $C$-test $S$ as
$L(S) = C \cap S$ and the {\em global} portion $G(S) = S \setminus C$.
\end{definition}


In order to prove Theorem~\ref{lem:minitest} below we need the following greedy algorithm.

\begin{center}
\fbox{~\begin{minipage}{0.909\textwidth}
\hspace{-0.09cm}{\bf Greedy-mini-test(${\cal T}$):}

\vspace{0.07cm}

\mbox{ }\hspace{0.2cm}\begin{minipage}{0.94\textwidth}
Start with ${\cal F}=\emptyset$. Add two tests
$T_i, T_j$ from ${\cal T}$ to ${\cal F}$ if this will increase the number of classes induced by ${\cal F}$
by at least $3$. Add a test $T_i$ from ${\cal T}$ to ${\cal F}$ if this will increase the number of classes
induced by $\cal F$ by at least $2$. Stop the construction if we reach $|{\cal F}| \geq 2k-2$.
\end{minipage}
\end{minipage}~}
\end{center}

\begin{lemma}\label{lem:greedyAlgorithm}
 If the algorithm Greedy-mini-test produces a set ${\cal F}$ with $|{\cal F}| \geq 2k-2$, then ${\cal F}$
is a $k$-mini test cover.
\end{lemma}

\begin{proof}
Observe that throughout Greedy-mini-test we have at least $\lceil\frac{3}{2} |{\cal F}|\rceil + 1$ classes, and when $|{\cal F}| \geq 2k-2$ then
$\lceil \frac{3}{2} |{\cal F}| \rceil + 1 \geq |{\cal F}| + k$. 
By construction we note that $|{\cal F}| \leq 2k-1 < 2k$, which implies that ${\cal F}$ is a $k$-mini test cover.~\qed
\end{proof}

\begin{theorem}\label{lem:minitest}
Suppose that ${\cal T}$ is a test cover for $[n]$. Then ${\cal T}$ contains a test cover of size at most $n-k$
if and only if ${\cal T}$ contains a $k$-mini test cover.
\end{theorem}

\begin{proof}
First suppose that ${\cal T}$ contains a $k$-mini test cover $\cal F$. Then by Lemma \ref{lem:oldClaim},
$\cal F$ can be extended to a test cover of size at most $n-k$.

Conversely, suppose ${\cal T}$ contains a test cover ${\cal F}'$ of size at most $n-k$. Now use algorithm 
Greedy-mini-test on ${\cal F}'$. If $|{\cal F}| \geq 2k-2$, where ${\cal F}$ is produced by Greedy-mini-test, 
then we are done by Lemma~\ref{lem:greedyAlgorithm},
so assume that $|{\cal F}| < 2k-2$. This implies that the following holds, as otherwise the algorithm wouldn't have terminated when it did.

\begin{enumerate}
 \item For every test $T_i \in {\cal F}' \backslash {\cal F}$, $T_i$ does not separate more than one class induced by ${\cal F}$.
 \item For every class $C$ induced by ${\cal F}$, and for every pair $T_i, T_j$ of $C$-tests in
${\cal F}'\backslash {\cal F}$, at least one of $(T_i\cap T_j)\cap C$, $(T_i \backslash T_j)\cap C$, $(T_j \backslash T_i)\cap C$ and $C \backslash(T_i \cup T_j)$ is empty.
\end{enumerate}

It can be seen that these properties hold even if we add one extra
test from ${\cal F}'\backslash {\cal F}$ to ${\cal F}$.

Therefore if we add $t$ tests from ${\cal F}' \backslash {\cal F}$, one at a time, this will subdivide a class $C$ into at most $t+1$ classes.
Furthermore, since each test separates at most one class, adding $t$ tests from ${\cal F}'\backslash {\cal F}$
to ${\cal F}$ will increase the number of classes induced by ${\cal F}$ by at most $t$.
It follows that ${\cal F}'$ induces less than $|{\cal F}|+k + |{\cal F}'\backslash {\cal F}| = |{\cal F}'|+k \le n$ classes. But this is a contradiction as ${\cal F}'$ is a test cover.
\qed \end{proof}

By Theorem \ref{lem:minitest} we get the following result, which allows us to concentrate on {\sc $k$-Mini Test Cover} in the next subsection.

\begin{corollary}\label{lem:equiv}
The problem $(n-k)$-{\sc Test Cover} is FPT if and only if {\sc $k$-Mini Test Cover} is FPT.
\end{corollary}

\subsection{Main Result}\label{sec:n-k}
We start with the following easy observation.
\begin{lemma}\label{lem:assumeSig}
 Let ${\cal T}$ be a test cover. Let ${\cal T^*}$ be the test cover formed from ${\cal T}$ by adding every singleton not already in ${\cal T}$.
 Then ${\cal T^*}$ has a $k$-mini test cover if and only if ${\cal T}$ also has a $k$-mini test cover.
\end{lemma}
\begin{proof}
 Assume ${\cal T^*}$ has a $k$-mini test cover ${\cal F}$. Form ${\cal F}'$ from ${\cal F}$ by removing all singletons. For each singleton removed the number of classes decreases by at most one. Hence, ${\cal F}'$ induces at least $|{\cal F}'| + k$ classes, and $|{\cal F}'|\le 2k$. Thus, ${\cal F}'$ is a $k$-mini test cover for ${\cal T}$.
The other direction is immediate since ${\cal T} \subseteq {\cal T^*}$.
\qed \end{proof}
\noindent
Due to Lemma \ref{lem:assumeSig}, \emph{hereafter we assume that every singleton belongs to ${\cal T}$.}

\smallskip

We will apply the algorithm Greedy-mini-test to find a collection
${\cal F} \subseteq {\cal T}$ of tests.
If $|{\cal F}| \geq 2k-2$ then we are done by Lemma~\ref{lem:greedyAlgorithm}, so for the rest of the arguments we assume that $|{\cal F}| < 2k-2$.
By construction, adding any new test to ${\cal F}$ increases the number of classes by at most $1$
and adding any two new tests to ${\cal F}$ increases the number of classes by at most $2$. 
Let the classes created by ${\cal F}$ be denoted by $C_1,C_2,\ldots,C_l$.
Note that $l \le 3k-2$.

\begin{lemma}\label{lem:ciset}
Any test $S \in {\cal T}\setminus {\cal F}$ cannot be a $C_i$-test and a $C_j$-test for $i \not= j$.
\end{lemma}
\begin{proof}
For the sake of contradiction, assume such a test $S$ exists. Then adding $S$ to $\cal F$ will increase the number of classes
by at least $2$, a contradiction to the definition of $\cal F$.
\qed \end{proof}

We may assume without loss of generality in the rest of this section that for all $C_i$-tests, $S$,
we have $|S \cap C_i| \leq |C_i|/2$.
Indeed, suppose $|S \cap C_i| > |C_i|/2$. Then we may replace $S$ in ${\cal T}$ with the test $S' = [n] \setminus S$. Observe that two items are separated by $S'$ if and only if they are separated by $S$, and so replacing $S$ with $S'$ produces an equivalent instance.
Furthermore, since $|S \cap C_i|>|C_i|/2$ we have that $|S'\cap C_i|\le |C_i|/2$.
Note that Lemma \ref{lem:ciset} still holds after replacing $S$ with $S'$, since for all $j \neq i$ either $S'\cap C_j= \emptyset$ or $C_j \subseteq S'$.


\begin{lemma} \label{strucL}
Any two $C_i$-tests $S,S' \in {\cal T}$ have either $L(S) \subseteq L(S')$ or $L(S') \subseteq L(S)$ or $L(S) \cap L(S') = \emptyset$.
\end{lemma}
\begin{proof}
For the sake of contradiction, assume $S,S'$ do not satisfy this condition. Then $C_i \cap (S \setminus S')$ is non-empty (otherwise, $L(S) \subseteq L(S')$).
Similarly, $C_i \cap (S' \setminus S)$ is non-empty. Since $L(S) \cap L(S') \neq \emptyset$,  $C_i \cap S \cap S'$ is non-empty. Finally observe that
$|L(S) \cup L(S')|=|L(S)|+|L(S')|-|L(S)\cap L(S')|\leq |C_i|/2+|C_i|/2-1 < |C_i|$. Hence $C_i \setminus (S \cup S')$ is non-empty.
Adding $S$ and $S'$ to $\cal F$ divides $C_i$ into four classes: $C_i \cap (S \setminus S')$, $C_i \cap (S' \setminus S)$, $C_i \cap (S \cap S')$ and $C_i \setminus (S \cup S')$.
This contradicts the maximality of $\cal F$.
\qed \end{proof}

An {\em
out-tree} $T$ is an orientation of a tree which has only one vertex of
in-degree zero (called the {\em root}); a vertex of $T$ of out-degree
zero is a {\em leaf}.

We now build an out-tree $O_i$ as follows. The root of the tree, $r\in V(O_i)$ corresponds to the set $C_i$. Each vertex $v \in V(O_i) \setminus {r}$ corresponds to a subset, $S_v \subseteq C_i$ such that there exists a $C_i$-test $S \in {\cal T}$ with $L(S)=S_v$.  Note that for a pair of vertices $u,v\in V(O_i)$ if $u \neq v$, then $S_u \neq S_v$. Add an arc from $v$ to $w$ in $O_i$ if
$S_w \subset S_v$ and there is no $u$ in $O_i$ with $S_w \subset S_u \subset S_v$.
By Lemma \ref{strucL} we note that $O_i$ is indeed an out-tree.

\begin{lemma} \label{degO}
  Every non-leaf in $O_i$ has out-degree at least two.
\end{lemma}

\begin{proof}
  Let $v$ be a non-leaf in $O_i$, and note that
$|S_v| \geq 2$. Let $w$ be any child of $v$ in $O_i$.
By definition there exists an item in $S_v \setminus S_w$ (as $|S_v|>|S_w|$), say $w'$.
As there is a singleton $\{w'\} \in {\cal T}$ there is a path from $v$ to $w'$ in $O_i$ and as $w' \not\in S_w$ the path does not use
$w$. Therefore $v$ has at least one other out-neighbour.
\qed \end{proof}

We now define the {\em signature} of a set $S' \subset C_i$ as follows.

$$Sig(S') = \{ G(S) :\  S \in {\cal T} \mbox{ and } L(S)=S' \}$$

\begin{lemma} \label{sig1}
 We have $|\{Sig(S') :\ S' \subset C_i \}| \leq 2^{2^{3k-1}}$.
\end{lemma}

\begin{proof}
Let ${\cal S}_i$ denote all sets, $S$, with $C_j \cap S = \emptyset$ or $C_j \subseteq S$ for all $j$ and furthermore $C_i \cap S = \emptyset$. Note that
$|{\cal S}_i|\le 2^{l-1} \leq 2^{3k-1}$, since $|\{C_1,\ldots,C_l\}~\setminus~\{C_i\}|  \leq 3k-1$.
Note that tests $U$ and $V$ with $L(U)=S'=L(V)$ have $G(U) \not= G(V)$, as $U\neq V$. Observe that all
$G(S)$ in $Sig(S')$ belong to ${\cal S}_i$ implying that there is at most $2^{|{\cal S}_i|} = 2^{2^{3k-1}}$ different choices for a signature.
\qed \end{proof}





\begin{lemma} \label{lem:main}
  There exists a function $f_1(k)$ such that either
the depth of the tree $O_i$ (i.e. the number of arcs in a longest path out of the root) is at most $f_1(k)$,
or in polynomial time, we can find a vertex $v$ in $O_i$ such that if there is a solution to our instance of $(n-k)$-{\sc Test Cover} then there is also a solution that does not use any test $S$ with
$L(S)=S_v$.

\end{lemma}

\begin{proof}
Let $f_1(k)=(32k-1) 2^{2^{3k-1}}$.
 Assume that the depth of the tree $O_i$ is more than $f_1(k)$ and let $p_0 p_1 p_2 \ldots p_a$ be a longest path in $O_i$ (so $a > f_1(k)$).
By Lemma \ref{sig1} and by the choice of $f_1(k)$, there is a sequence $p_{j_1},p_{j_2},\ldots, p_{j_{32k}}$, where $1 \leq j_1 < j_2 < \cdots < j_{32k} \leq a$ and
all sets corresponding to $p_{j_1},p_{j_2},\ldots, p_{j_{32k}}$ have the same signature.

Let $S^*$ be the set  corresponding to $p_{j_{16k}}$.
We will show that if there is a solution to our instance of $(n-k)$-{\sc Test Cover} then there is also a solution that does not use any test $S$ with
$L(S)=S^*$.

Assume that there is a solution to our instance of $(n-k)$-{\sc Test Cover} and assume that we pick a solution ${\cal T}'$ with as few tests, $S$, as possible
with $L(S)=S^*$. For the sake of contradiction assume that there is at least one test $S'$ in our solution with $L(S')=S^*$.
By Theorem \ref{lem:minitest} there is a $k$-mini test cover, ${\cal F}'$,  taken from ${\cal T}'$.
Initially let ${\cal F}'' = {\cal F}'$.
While there exists a vertex $r \in C_q$ and $r' \in C_p$ ($q \not= p$) which are not separated by ${\cal F}''$
then add any test from ${\cal F}$
which separates $r$ and $r'$ to ${\cal F}''$
(recall that ${\cal F}$
is the test collection found by Greedy-mini-test).
 Note that this increases the size of ${\cal F}''$ by $1$ but also increases the number of classes induced by ${\cal F}''$  by at least $1$.
We continue this process for as long as possible. As ${\cal F}'' \subseteq {\cal F} \cup {\cal F}'$ we note that $|{\cal F}''| \leq 2k + 2k = 4k$.
Furthermore, by construction, vertices in different $C_j$'s are separated by tests in ${\cal F}''$. Also note that the number of classes induced by
${\cal F}''$ is at least $|{\cal F}''|+k$ (as the number of classes induced by
${\cal F}'$ is at least $|{\cal F}'|+k$).

For every test, $S$, in ${\cal F}''$ color the vertex in $O_i$ corresponding to $L(S)$ blue.
For every vertex, $v \in V(O_i)$, color $v$ red if all paths from $v$ to a leaf in $O_i$ use at least
one blue vertex and $v$ is not already colored blue. Finally for every vertex, $w \in V(O_i)$, color $w$ orange if
all siblings of $w$ (i.e. vertices with the same in-neighbour as $w$) are colored blue or red and $w$ is not colored blue or red.
We now need the following:\\

\noindent{\bf Claim A:} The number of colored vertices in $O_i$ is at most $16k-2$.\\
{\em Proof of Claim A:}  As $|{\cal F}''| \leq 4k$ we note that the number of blue vertices is at most $4k$. We will now show that the number of red vertices is at most $4k-1$.
Consider the forest obtained from $O_i$ by only keeping arcs out of red vertices. Note that any tree in this forest has all its leaves colored
blue and all its internal vertices colored red. Furthermore, by Lemma \ref{degO} the out-degree of any internal vertex is at least $2$. This implies that
the number of red vertices in such a tree is less than the number of blue vertices. As this is true for every tree in the forest we conclude that the number of
red vertices in $O_i$ is less than the number of blue vertices in $O_i$ and is therefore bounded by $4k-1$.

We will now bound the number of orange vertices. Since every orange vertex in $O_i$ has at least one sibling colored blue or red
(by Lemma \ref{degO}).
and any blue or red vertex can have at most one orange sibling we note that the number of orange vertices cannot be more than the number of vertices
colored blue or red. This implies that the number
of orange vertices is at most $8k-1$.\\

By Lemma \ref{lem:oldClaim}, we note that some test, $S^x$, in ${\cal F}''$ has $L(S^x)=S^*$ (as otherwise extend ${\cal F}''$ by singletons
to a test cover where no test, $S$, in the solution has $L(S)=S^*$,
a contradiction to our assumption).
Now create ${\cal F}^x$ as follows. Initially let ${\cal F}^x$ be obtained from ${\cal F}''$ by removing the test $S^x$.
Let $p_{j_{i'}}$ be an uncolored vertex in $\{p_{j_1},p_{j_2},\ldots, p_{j_{16k-1}}\}$
and let
$p_{j_{i''}}$ be an uncolored vertex in $\{p_{j_{16k+1}},p_{j_{16k+2}},\ldots, p_{j_{32k-1}}\}$ (note that we do not pick $p_{j_{32k}}$).
Let $S_1^x$ be a test in ${\cal T}$ with $G(S_1^x)=G(S^x)$ and $L(S_1^x)$ corresponding to the vertex $p_{j_{i'}}$
and let $S_2^x$ be a test in ${\cal T}$ with $G(S_2^x)=G(S^x)$ and $L(S_2^x)$ corresponding to $p_{j_{i''}}$.
These tests exist as the signature of all sets corresponding to vertices in $p_{j_1},p_{j_2},\ldots, p_{j_{32k}}$ are the same.
Now add $S_1^x$ and $S_2^x$ to ${\cal F}^x$. The following now holds.

\smallskip

\noindent{\bf Claim B:} The number of classes induced by
${\cal F}^x$ is at least $|{\cal F}^x|+k$.

\noindent{\em Proof of Claim B:} Let $u,v \in [n]$ be arbitrary. If $u,v \not\in C_i$ and they
are separated by ${\cal F}''$, then they are also separated by ${\cal F}^x$, as if
they were separated by $S^x$ then they will now be separated by $S_1^x$ (and $S_2^x$). Now assume that $u \in C_i$ and $v \not\in C_i$. If $u \in L(S^x)$ and $u$ and $v$
were separated by $S^x$ then they are also separated by $S_1^x$. If $u \not\in L(S^x)$ and $u$ and $v$
were separated by $S^x$ then they are also separated by $S_2^x$. So as $u$ and $v$ were separated by ${\cal F}''$ we note that they are also separated by ${\cal F}^x$.
We will now show that the number of classes completely within $C_i$ using ${\cal F}^x$ is at least one larger than when using ${\cal F}''$.

By
Lemma \ref{strucL}
we note that deleting $S^x$ from ${\cal F}''$ can decrease the number of classes within $C_i$ by at most one (it may decrease the number of classes
in $[n]$ by more than one). We first show that adding the test $S_1^x$ to ${\cal F}'' \setminus \{S^x\}$ increases the number of classes within $C_i$ by at least one.
As $p_{j_{i'}}$ is not colored there is a path from $p_{j_{i'}}$ to a leaf, say $u_1$, without any blue vertices. Furthermore as $p_{j_{i'}}$ is not orange we note that it has a
sibling, say $s'$, that is not colored and therefore has a path to a leaf, say $u_2$, without blue vertices. We now note that $u_1$ and $u_2$ are not separated in ${\cal F}''$ (and therefore
in ${\cal F}'' \setminus \{S^x\}$). However adding the test $S_1^x$ to ${\cal F}'' \setminus \{S^x\}$ does separate $u_1$ and $u_2$ (as $u_1 \in S_1^x$ but $u_2 \not\in S_1^x$). Therefore
the classes within $C_i$ has increased by at least one by adding  $S_1^x$ to ${\cal F}'' \setminus \{S^x\}$.

Analogously we show that adding the test $S_2^x$ to ${\cal F}'' \cup \{S_1^x\} \setminus \{S^x\}$ increases the number of classes within $C_i$ by at least one.
As $p_{j_{i''}}$ is not colored there is a path from $p_{j_{i''}}$ to a leaf, say $v_1$, without blue vertices. Furthermore as $p_{j_{i''}}$ is not orange we note that it has a
sibling, say $s''$, that is not colored and therefore has a path to a leaf, say $v_2$, without blue vertices. We now note that $v_1$ and $v_2$ are not separated in ${\cal F}''$ (and therefore
in ${\cal F}'' \cup \{S_1^x\} \setminus \{S^x\}$, as $p_{j_{i'}}$ lies higher in the tree $O_i$ and therefore the test $S_1^x$ does not separate $u$ and $v$).
However adding the test  $S_2^x$ to ${\cal F}'' \cup \{S_1^x\} \setminus \{S^x\}$ does separate $v_1$ and $v_2$
(as $v_1 \in S_2^x$ but $v_2 \not\in S_2^x$). Therefore
the classes within $C_i$ has increased by at least one by adding  $S_2^x$ to ${\cal F}'' \cup \{S_1^x\} \setminus \{S^x\}$.
So we conclude that the number of classes within $C_i$ has increased by at least one and as any vertex not in $C_i$ is still separated from exactly the same vertices in ${\cal F}^x$ as
it was in ${\cal F}''$ we have proved Claim B.
\\

By Lemma \ref{lem:oldClaim} and Claim B we get a solution with fewer tests, $S$, with $L(S)=S^*$, a contradiction.
\qed \end{proof}

Suppose the depth of $O_i$ is greater than $f_1(k)$, and let $S^*$ be the set found by the above lemma.
Then we can delete all tests, $S$, with $L(S)=S^*$ from ${\cal T}$ without changing the problem, as if there is a solution for the instance then there is one that does not contain any test $S$ with $L(S)=S^*$.
Therefore we may assume that the depth of $O_i$ is at most $f_1(k)$.

\begin{lemma} \label{lem:main2}
  There exist functions $f_2(d,k)$ and $f_3(d,k)$, such that in polynomial time we can reduce $([n],{\cal T}, k)$ to an instance
such that the following holds for all vertices $v \in O_i$, where
$d$ is the length (i.e. number of arcs) of a longest path out of $v$ in $O_i$:
(1) $N^+(v) \leq f_2(d,k)$ and (2) $|S_v| \leq f_3(d,k)$.
\end{lemma}
\begin{proof}
  Let $v$ be a vertex in $O_i$ and let $d$ be the length of a longest path out of $v$ in $O_i$. We will prove the lemma by induction on $d$.
If $d=0$ then $v$ is a leaf in $O_i$ and $N^+(v)=0$ and $|S_v| =1$ (as all singletons exist in ${\cal T}$). So now assume that $d \geq 1$ and
the lemma holds for all smaller values of $d$. We note that the way we construct $f_3(d,k)$ below implies that it is increasing in $d$.

We will first prove part (1). Let $N^+(v) = \{w_1,w_2,w_3, \ldots , w_b\}$ and note that $|S_{w_j}| \leq f_3(d-1,k)$ for all $j=1,2,\ldots,b$ (by
induction and the fact that $f_3(d,k)$ is increasing in $d$).
Let $Q_j$ be the subtree of $O_i$ that is rooted at $w_j$ for all $j=1,2,\ldots,b$. As part (1) holds for all vertices in $Q_j$ we note that
there are at most $g(d,k)$ non-isomorphic trees in $\{Q_1,Q_2,\ldots,Q_b\}$ for some function $g(d,k)$.
Furthermore the number of vertices in each $Q_j$ is bounded by $2 f_3(d-1,k) -1$ by Lemma \ref{degO} and induction (using part (2) and the
fact that every leaf in $Q_j$ corresponds to a singleton in $C_i$ and the number of leaves are therefore bounded by $ f_3(d-1,k)$).
By Lemma \ref{sig1} the number of distinct signatures is bounded by $2^{2^{3k-1}}$. Let $f_2(d,k)$ be defined as follows.

$$ f_2(d,k) = 2k \cdot  g(d,k) \left[ 2^{2^{3k-1}} \right]^{2f_3(d-1,k)-1} $$

So if $b> f_2(d,k)$ there exists at least $2k+1$ trees in  $\{Q_1,Q_2,\ldots,Q_b\}$ which are strongly isomorphic, in the sense that a one-to-one mapping from one to the other
maintains arcs as well as signatures (a vertex with a given signature is mapped into a vertex with the same signature).
Without loss of generality assume that $Q_1$ is one of these at least $2k+1$ trees.
We now remove all vertices in $Q_1$ as well as all tests $S$ with $L(S)$ corresponding to a vertex in $Q_1$. Delete all the items in $S_{w_1}$.
Let the resulting test collection be denoted by ${\cal T'}$, and denote the new set of items by $[n']$.
We show this reduction is valid in the following claim.\\

\noindent{\bf Claim:} This reduction is valid (i.e. $([n] ,{\cal T}, k)$ and $([n'], {\cal T'}, k)$ are equivalent).

\noindent{\em Proof of Claim:}
Observe that any $k$-mini test cover in ${\cal T}'$ is a $k$-mini test cover in ${\cal T}$, and so $([n], {\cal T}, k)$ is a {\sc Yes}-instance if $([n'],{\cal T}', k)$ is a {\sc Yes}-instance.

For the converse, assume ${\cal T}$ contains a $k$-mini test cover ${\cal F}'$, and for each test $S$ in ${\cal F}'$, color the vertex in $O_i$ corresponding to $L(S)$ blue.
We first show we may assume $Q_1$ is uncolored. For suppose not, then since $|{\cal F}'|\le 2k$, then some other tree $Q_j$ that is strongly isomorphic to $Q_1$ is uncolored. In this case, we may replace the tests $S$ in ${\cal F}'$ with $L(S)$ corresponding to a vertex in $Q_1$, by the equivalent tests $S'$ with $L(S')$ corresponding to a vertex in $Q_j$.

So assume $Q_1$ is uncolored. Then ${\cal F}'$ is still a subcollection in ${\cal T'}$. It remains to show that ${\cal F}'$ still induces
at least $|{\cal F}'|+k$ classes over $[n']$.
Observe that this holds unless there is some class $C$ induced by $|{\cal F}'|$ that only contains items from $S_{w_1}$. But this can only happen if some item in $S_{w_1}$ is separated from $S_{w_j}$ by a test in ${\cal F}'$, for all $j \in \{2, \dots b\}$. But since $b > |{\cal F}|+1$, there exists $j \neq 1$ such that $Q_j$ is not coloured. Then since $w_1, w_j$ are siblings, no test in ${\cal F}'$ can separate $S_{w_1}$ from $S_{w_j}$. Thus ${\cal F}'$ induces at least $|{\cal F}'|+k$ classes over $[n']$, and so
${\cal F}'$ is still a $k$-mini test cover in the new instance.
%
%
%
%
Thus, $( [n'],{\cal T}',k)$ is a {\sc Yes}-instance if and only if $([n], {\cal T}, k)$ is a {\sc Yes}-instance.
\\

By the above claim we may assume that $b \leq f_2(d,k)$, which proves part (1).
 We will now prove part (2).
As we have just proved that $b \leq f_2(d,k)$ and $|S_{w_j}| \leq f_3(d-1,k)$ for all $j=1,2,\ldots,b$, we note that
(2) holds with $f_3(d,k) = f_3(d-1,k) \times f_2(d,k)$.
\qed \end{proof}

\begin{theorem}
 The $(n-k)$-{\sc Test Cover} problem is fixed-parameter tractable.
\end{theorem}

\begin{proof}
Given a test cover, construct a subcollection ${\cal F}\subseteq {\cal T}$ using algorithm Greedy-mini-test. If $|{\cal F}|\ge 2k$, the instance is a {\sc Yes}-instance since ${\cal F}$ induces at least $\frac{3}{2}|{\cal F}|$ classes.
Otherwise, $|{\cal F}|< 2k$ and ${\cal F}$ induces at most $3k$ classes.
By Lemma \ref{lem:main}, we may assume that $O_i$ has depth at most $d=f_1(k)$, and by Lemma \ref{lem:main2} part (2) we may assume that $|C_i| \le f_3(d,k)$, for each class $C_i$ induced by ${\cal F}$.
Thus $|C_i| \leq f_3(f_1(k),k)$.

Hence there are at most $3k$ classes, the size of each bounded by a function of $k$, so the number of items in the problem is bounded by a function of $k$. Thus, the problem can be solved by an algorithm of runtime depending on $k$ only.
\qed
 \end{proof}

\section{Conclusion}\label{sec:con}

We have considered four parameterizations of {\sc Test Cover}  and established their parameterized complexity.
The main result is fixed-parameter tractability of  $(n-k)$-{\sc Test Cover}. Whilst it is a positive result, the runtime of the algorithm that we can obtain is not practical and we hope that subsequent improvements of our result can bring down the runtime to a practical level. Ideally, the runtime should be $c^k(n+m)^{O(1)}$, but it is not always possible \cite{LokMarSau11}.

\paragraph{Acknowledgment}
This research was partially supported by an International Joint grant of Royal Society.

\end{document}